\documentclass[11pt,a4paper,oneside]{amsart}

\usepackage[latin1]{inputenc}
\usepackage{amsmath}
\usepackage{amsfonts}
\usepackage{amssymb}
\usepackage{graphicx}
\usepackage{amsthm}
\usepackage{enumerate,hyperref} 
\usepackage[version=3]{mhchem}
\usepackage[all]{xy}
\usepackage{adjustbox}
\usepackage{arydshln}
\usepackage{multirow}
\usepackage{resizegather}
\usepackage[dvips,pdftex]{color}
\usepackage{pgf,tikz}
\usepackage{url}
\usepackage[left=3cm,right=4cm]{geometry}
\usepackage{textcomp}
 
\usepackage[sort&compress,comma,square,numbers]{natbib}

\theoremstyle{plain}
\newtheorem{thm}{Theorem}[section]

\newtheorem{corollary}[thm]{Corollary}
\newtheorem{lemma}[thm]{Lemma}

\theoremstyle{definition}
\newtheorem{definition}[thm]{Definition}

\newtheorem{remark}[thm]{Remark}
\newtheorem{example}[thm]{Example}

\newcommand{\cN}{\mathcal{N}}
\newcommand{\cS}{\mathcal{S}}
\newcommand{\cC}{\mathcal{C}}
\newcommand{\cR}{\mathcal{R}}
\newcommand{\cY}{\mathcal{Y}}
\newcommand{\cX}{\mathcal{X}}

\newcommand{\R}{\mathbb{R}}
\newcommand{\Z}{{\mathbb Z}}

\renewcommand{\k}{\kappa}

\DeclareMathOperator{\im}{Im}
\DeclareMathOperator{\LT}{LT}
\DeclareMathOperator{\LM}{LM}
\DeclareMathOperator{\Rem}{rem}
\DeclareMathOperator{\Span}{span}
\DeclareMathOperator{\lex}{lex}

\begin{document}

\title{Gr\"obner bases of reaction networks with intermediate species}
\author{AmirHosein Sadeghimanesh$^1$, Elisenda Feliu$^{1,2}$}
\date{\today}

\footnotetext[1]{Department of Mathematical Sciences, University of Copenhagen, Universitetsparken 5, 2100 Copenhagen, Denmark.}
\footnotetext[2]{Corresponding author: efeliu@math.ku.dk}

	\maketitle
	\begin{abstract}
		In this work we consider the computation of  Gr\"{o}bner bases of the steady state ideal of reaction networks equipped with mass-action kinetics. Specifically, we focus on the role of \emph{intermediate} species and the relation between the extended network (with intermediate species) and the core network (without intermediate species). 

		We show that  a Gr\"{o}bner basis of the steady state ideal of the core network always lifts to a  Gr\"{o}bner basis of the steady state ideal of the extended network by means of linear algebra, with a suitable choice of monomial order. As illustrated with examples, this contributes to a substantial reduction of the computation time, due mainly to the reduction in the number of variables and polynomials. We further show that if the steady state ideal of the core network is binomial, then so is the case for the extended network, as long as an extra condition is fulfilled. For standard networks, this extra condition can be visually explored from the network structure alone.


	\end{abstract}
	
	\section*{Introduction}
	
 	Parametric polynomial systems of equations arise in the natural sciences when modeling ecosystems, cell behavior, the spread of an illness, and molecular interactions within the cell, to name a few examples. In these scenarios questions of interest often boil down to describing the solutions to these systems   for varying values of the parameters. Although only non-negative solutions are typically meaningful, the standard tool in computational algebraic geometry to study algebraic varieties, namely  Gr\"{o}bner bases, has proven useful.
		  However, due to the parametric coefficients, the computation of a  reduced Gr\"{o}bner basis can be time consuming for realistic examples, which typically involve many variables and parameters.  The computation time depends mainly on the degree of the polynomials, the number of variables and coefficients, 
 the choice of the monomial order and the  used method \cite{Cox,F4,p-modular-Winkler,F5Complexity,EliminationOrders}. 
 These universal considerations target generic polynomial systems, but, in applications, the 
 structure of the particular system of interest might favor one method or one monomial order over another.

We focus on a specific type of polynomial systems that arise when modeling chemical reaction networks with mass-action kinetics  \cite{feinbergnotes,gunawardena-notes}. 
Specifically,  the evolution of the concentrations of the  species of a chemical reaction network in time is described under mass-action 
by a system of ordinary differential equations in $\R^n$
\[ \tfrac{dx_i}{dt}=f_{\k,i}(x),\qquad i=1,\dots,n\]
with $f_{\k,i}(x)$ polynomial. The monomials of each $f_{\k,i}(x)$  depend  on the reaction network structure alone, and the coefficients  depend on the \emph{reaction rate constants} $\k$, which are often unknown and thus treated as parameters.  The steady states, or equilibrium points, of the system are  the \emph{non-negative} points of the variety defined by the \emph{steady state ideal} $I_\k=\langle f_{\k,1}(x),\dots,f_{\k,n}(x)\rangle$.

 The question of restricting to non-negative steady states remains challenging and no straightforward solutions have been proposed. Despite of this, Gr\"{o}bner bases have been for example used for model discrimination \cite{Wnt-Matroid,InvariantModelDiscrimination,ComplexLinearInvariantsGunawardeena,GeometryOfMultisitePhosphorylationGunawardeena,DistributiveProcessiveMultisitePhosphorylationGunawardeena}. They are also used to decide whether the steady state ideal is binomial, that is, whether any reduced Gr\"{o}bner basis consists of polynomials with at most two terms. If this is the case, then methods to detect the existence of multiple steady states can be
applied   \cite{Dickenstein-Toric,Feliu-Sign}. 

In this work we exploit the specific structure of the steady state ideal, which reflects the structure of the reaction network, to guide the selection of good monomial orders and to compute  reduced Gr\"{o}bner bases faster. Specifically, we consider a frequent and nicely-behaved class of species introduced in \cite{Feliu-Simplifying} called \emph{intermediate species} (or intermediates, for short). Intermediates give rise to linear terms in the steady state polynomials, and they can be removed from a reaction network resulting in a smaller \emph{core network} with only the non-intermediates. A key property is that steady states of the core network can be lifted to steady states of the \emph{extended network}.

 The first main result of this work is Theorem \ref{Proposition Groebner Intermediates}, where we show how to obtain a 
	Gr\"{o}bner basis of the extended network from one of the core network using linear algebra. The result implicitly gives  good monomial orders, namely, those for which the concentration of the intermediates are larger than for the non-intermediates, and are lexicographic in the variables corresponding to the intermediates. Example \ref{Example Conradi System 169} illustrates  the computational advantage of using our approach compared with other methods. Additionally, we conclude that the analysis of the steady state ideal of the core network is sufficient for model discrimination.
	
	 The second main result, Theorem \ref{Theorem Binomiality and Intermediates},  addresses how to decide whether the steady state ideal is binomial. We show that if the steady state ideal of the core network is binomial, then this is also the case for the steady state ideal of the extended network provided an extra condition is fulfilled. In typical networks, this extra condition   can be readily checked from the network structure alone. 	 When the core network has a homogeneous steady state ideal (which happens frequently for realistic reaction networks), then one can employ the  linear algebra-based method introduced in  \cite{Conradi-Binomial} to detect whether the steady state ideal of the core network is binomial. Then, combined with our result, we obtain a faster method to address 
whether the steady state ideal of the original network is binomial, which  does not rely on the computation of a Gr\"{o}bner basis. 
	 
The key property behind our results is that intermediates define a square linear subsystem of full rank among the steady state polynomials. Its solution and posterior substitution into the remaining polynomials gives rise to a smaller ideal in the non-intermediates. A Gr\"{o}bner basis of the small ideal can then be lifted to a Gr\"{o}bner basis of the original ideal. Our approach can be theoretically applied to arbitrary parametric ideals, after detection of linear subsystems among a set of generators. However, technical conditions that are necessary for our results to hold might not be straightforward to check, since we overcome this difficulty by exploiting the network structure.

	The structure of the paper is as follows. We start by introducing reaction networks and basic concepts such as the steady state ideal. Intermediates are introduced in Section \ref{sec:intermediates}. In Section \ref{sec:groebner} we address  Gr\"{o}bner bases of networks with intermediates, discuss binomial steady state ideals and relate our work to \cite{Conradi-Binomial}.
	 In Section \ref{sec:algebraic-independent}  a technical condition of algebraic independence of a set of rational functions,  which is assumed in the former sections, is discussed.  Finally, in the last section, we discuss another class of special species, namely enzymes, that might lead to similar results concerning the computation of Gr\"{o}bner bases.

	\section{The steady state ideal of a reaction network}\label{sec:steady_state_ideal}
	
	We follow the formalism of \cite{Feliu-Simplifying}. See also \cite{feinbergnotes,gunawardena-notes} for an introduction to reaction networks.
	Subscripts $\geq 0, >0$ on $\R$ (resp. $\Z$) refer to the non-negative and positive real numbers (resp. integer numbers).
	
	A \emph{reaction network} is an ordered triple $\cN=\left(\cS,\cC,\cR\right)$ where $\cS$, $\cC$ and $\cR$ are three sets called the  set of \emph{species}, \emph{complexes} and \emph{reactions}, respectively. Here $\cS$ is a finite set and $\cC$ is a finite set of linear combinations of elements of $\cS$ with coefficients in $\mathbb{Z}_{\geq 0}$. A \emph{reaction} is an ordered pair of complexes $\left(c,c'\right)$ in $\cC^2$,  usually denoted as $c\rightarrow c'$. For the reaction $c\rightarrow c'$, the complex $c$ is called the \textit{reactant}   and $c'$ is called the \textit{product}.
	
	A digraph is associated with a reaction network as follows.  The vertex set is $\cC$ and there is a directed edge from the reactant to the product of every reaction. If both reactions $c\rightarrow c'$ and $c'\rightarrow c$ for two complexes $c$ and $c'$ exist, then the notation $c\rightleftharpoons c'$ is used and the reaction is said to be reversible.
	
Complexes that are not part of any reaction or species that are not part of any complex do not appear in the digraph. Therefore, the reaction network cannot uniquely be determined from the digraph alone. For simplicity, however,  we often introduce a reaction network by its digraph and implicitly assume that the set of complexes equals the set of vertices and the set of species consists of the species that appear in at least one complex.

Write $\cS=\{X_1,\dots,X_n\}$, such that the set of species is implicitly ordered. Then a complex $c$ is of the form $c_1X_1+\dots+c_nX_n$, which we also write in vector form as $c=(c_1,c_2,\dots,c_n)\in \mathbb{Z}_{\geq 0}^n$. 
With this representation, $c_i$ is called the \emph{stoichiometric coefficient} of $X_i$ in $c$. 
	
	\begin{example}\label{Example Reaction Network}
		Let $\cS=\{X_1,X_2,X_3,X_4\}$, $\cC=\{X_1+X_3,X_4,X_2+X_3\}$, $\cR=\{X_1+X_3\rightarrow X_4,X_4\rightarrow X_1+X_3,X_4\rightarrow X_2+X_3\}$. The network $\cN=\left(\cS,\cC,\cR\right)$ is represented with the following digraph
		\[X_1 + X_3\ce{ <=>} X_4 \ce{-> } X_2 + X_3.\]
		The complexes $X_1+X_3$ and $X_4$ appear both as reactants and products while $X_2+X_3$ appears only as a product.
	\end{example}

	We next construct a system of Ordinary Differential Equations (ODEs) that models the variation of the concentration of each species in time and introduce the relevant polynomials $F_i(x)$ that are the focus of this work.
	We denote the concentration of each species $X_i$ in lower-case $x_i$. 
For each reaction $c\rightarrow c'$, we introduce a  parameter $k_{c\rightarrow c'}$, and a polynomial $F_i(x)$ is associated with every species $X_i$  as follows:
	\begin{equation}\label{eq:Fix}
	F_i(x)=\sum_{c\rightarrow c'\in\cR}(c_i'-c_i)k_{c\rightarrow c'} \, x^c \ \in\ \mathbb{R}(k)[x],
	\end{equation}
	where $x^c= x _1^{c_1}\dots x_n^{c_n}$. 
	Here $x=(x_1,\dots,x_n)$ and $\mathbb{R}(k)$ is the field of rational functions  with variables $k_{c\rightarrow c'}$ and real coefficients. The symbol $k$ stands for  the parameter vector \[k=(k_{c\rightarrow c'} \mid c\rightarrow c'\in\cR). \]
	
	For a chosen positive value $k^\star\in\mathbb{R}^\cR_{>0}$ of the parameter  vector, we let $F_{k^\star,i}(x)\in\mathbb{R}[x]$ denote the image of $F_i(x)$ under the evaluation map 
	\[\mathbb{R}(k)\rightarrow\mathbb{R},\qquad k_{c\rightarrow c'} \mapsto k_{c\rightarrow c'}^\star.\]

With this choice of $k^\star$, the ODE system of the reaction network under \emph{mass-action kinetics} is
\begin{equation}\label{eq:ODE}
\dot{x}_i=F_{k^\star,i}(x),\quad i=1,\dots,n,\qquad x\in \R^n_{\geq 0}.
\end{equation}
The value $k^\star_{c\rightarrow c'} >0$ is called  the \emph{reaction rate constant}  of $c\rightarrow c'$ and is usually depicted as a label of the reaction in the associated digraph. 
	By \cite{Sontag:2001}, if the starting condition of \eqref{eq:ODE} belongs to $\R_{>0}^n$ (resp. $\R_{\geq 0}^n$), then so does the trajectory for all positive times in the interval of definition.
	
	The \emph{steady states} of the network are the common zeros of $F_{k^\star,i}(x)$, $i=1,\dots,n$.
	In applications, only non-negative real solutions have meaning and mostly, positive steady states are interesting, meaning all concentrations are positive.
	Since the values of the reaction rate constants are in general unknown, they are treated as parameters of the system. Thus we aim at studying the zeros of the system of polynomials $F_{i}(x)=0$, for $i=1,\dots,n$ in $\R(k)$ and specially the positive zeros after specifying values for $k$.
	
	\begin{definition}\label{Definition1.1}
		Let $\cN=(\cS,\cC,\cR)$ be a reaction network with $\cS=\{X_1,\dots,X_n\}$.
		\begin{enumerate}[(a)]
			\item $F_{i}(x)\in\mathbb{R}(k)[x]$ is called the \emph{steady state polynomial} of $X_i$.
			\item The ideal generated by the steady state polynomials of all the species in the network in the ring $\mathbb{R}(k)[x]$ is called the \emph{steady state ideal} of the network:
			\[I_\cN=\big\langle F_{i}(x)\mid i=1,\dots,n\big\rangle \ \subseteq \mathbb{R}(k)[x].\]
		\end{enumerate}
	\end{definition}

The set of steady states for a vector of reaction rate constants $k^\star$ is thus the solution set to any basis (set of generators) of $I_{\cN}$ specialized	to  $k^\star$. 
	
	It follows from \eqref{eq:Fix} and \eqref{eq:ODE}  that  for all $x\in \R^n$, the vector $$F_k(x)=(F_{k,1}(x),\dots,F_{k,n}(x))$$ lies in the vector subspace  $S =  \Span(  \{c-c' \mid c\rightarrow c'\in\cR\}) \subseteq \R^n$. If $s=\dim(S)$, then $n-s$ of the steady state polynomials can be written as linear combinations of the remaining $s$ polynomials.  We conclude that it is always possible to find a basis of $I_\cN$ with cardinality $\dim(S)$. 
	
	\begin{example}\label{Example Steady state system}
		(continued from Example \ref{Example Reaction Network})  The ODE system of the reaction network with digraph
		\[X_1 + X_3 \ce{<=>[k_1][k_2]} X_4 \ce{->[k_3]} X_2 + X_3\]
		is
\begin{align*}
		\dot{x}_1 & =  -k_1x_1x_3+k_2x_4 & 
		\dot{x}_2 & =  k_3x_4\\
		\dot{x}_3 & =  -k_1x_1x_3+k_2x_4+k_3x_4 & 
		\dot{x}_4 & =  k_1x_1x_3-k_2x_4-k_3x_4.
		\end{align*}
In this case $\dim(S)=2$, $k=(k_1,k_2,k_3)$ and the steady state ideal is
		\[I_\cN=\big\langle -k_1x_1x_3+k_2x_4,k_3x_4\big\rangle \  \subseteq \R(k)[x].\]
	\end{example}

		\section{Intermediates and steady states}\label{sec:intermediates}
	
In this subsection we introduce a special type of species of interest: intermediates.

\begin{definition}	\label{def:intermediate}
We say that $\cY\subseteq \cS$ is a subset of \emph{intermediates} if 
each $Y\in \cY$ fulfills:
\begin{itemize}
\item   $Y\in\cC$ and the stoichiometric coefficient of $Y$ in all other complexes is zero, and 
\item there exists at least one reaction having $Y$ as reactant and at least one reaction having $Y$ as product. 
\end{itemize}
Each $Y\in \cY$ is called an \emph{intermediate}.
\end{definition}

Whenever a set of intermediates $\cY$ is given, we partition  the set of species  into two disjoint subsets $\cY=\{Y_1,\dots,Y_m\}$ and $\cX=\{X_1,\dots,X_n\}$ of non-intermediates. We assume further that the set of species is ordered such that the species $Y_1,\dots, Y_m$ are first. With this convention, we let 
	$(y,x)$ denote the concentration vector of  all species: $x$ is the concentration vector of the species in $\cX$ and $y$ of the species in $\cY$. A complex is either an intermediate in $\cY$ or it contains  only non-intermediates. In the latter case we say that $c$ is a \emph{non-intermediate complex}.

Note that given $\cY$, we refer to the intermediates  of the network as the species in $\cY$, even though there might be other species in $\cX$, regarded as non-intermediates, that fulfill the two items in Definition~\ref{def:intermediate}.

	\begin{example}\label{Example Popular Intermediates}
The most common mechanism involving intermediates  is of the following form:
\[X+E\ce{->}Y\ce{->} X'+E\]
or variations of it by letting one or both reactions being reversible.
\emph{Isomerism mechanisms} among intermediates are also common: 
\[
	Y\ce{->} Y' \qquad 	Y\ce{<=>} Y'.
\]
Combination of these mechanisms yields to more elaborate networks involving intermediates, as in Examples~ \ref{Example MAPK mu} and \ref{Example Conradi System 169} below.
\end{example}

	\begin{definition}\label{Definition Intermediate} Let $\cY$ be a set of intermediates and $Y\in \cY$. 
		\begin{itemize}
			\item A non-intermediate complex $c$ is called an \emph{input} for $Y$ if there is a directed path from $c$ to $Y$ in the digraph associated with the network, such that all vertices other than $c$ belong to $\cY$.
			\item $Y$ is called an $\ell$-input intermediate if there are $\ell$ inputs for $Y$.
		\end{itemize}
	\end{definition}

	\begin{example}
Consider the following network with $\cY=\{Y_1,Y_2,Y_3\}$:
	\[X_1 + X_2 \ce{->}  Y_1 \ce{<=>} Y_2 \ce{<=>} Y_3 \ce{->}  X_3 + X_4.\]
	There are two non-intermediate complexes, $X_1+X_2$ and $X_3+X_4$. 
	The species $Y_1$, $Y_2$, $Y_3$ are all 1-input intermediates. Note that  $Y_2$ is however the product of two reactions.
	
	Consider now the following network with $\cY=\{Y\}$:
	\[X_1 + X_2 \ce{ <=>} Y \ce{<=>} X_3 + X_4.\]
	The species $Y$ is a 2-input intermediate and $X_1+X_2$ and $X_3+X_4$ are both inputs for $Y$.
	\end{example}

\subsection{Intermediates and steady states}
	Let $\widetilde{\cN}$ be a reaction network with  a set of intermediates $\cY=\{Y_1,\dots,Y_m\}$. Consider the steady state polynomials of the intermediates and denote the parameter vector of reaction rate constants by $\kappa$ (the reason why will be made clear below). 
By definition, for every intermediate $Y_i$, the variable $y_i$  is only part of the monomial $y_i$  in \eqref{eq:Fix}.  Thus, the system with $m$ equations
	\[F_{1}(y,x)=\dots=F_{m}(y,x)=0\]
	is linear in $y_1,\dots,y_m$. 
	It is shown in \cite{Feliu-Simplifying} that this system has a unique solution for fixed positive values of $\k$ and $x$, which is further positive.
The solution is of the form
\[y_i=\sum_{c\in\cC}\mu_{i,c}\, x^c,\qquad \textrm{where }\quad  \mu_{i,c}\in\mathbb{R}_{\geq 0}(\k), \quad i=1,\dots,m.\]
The explicit dependence of $\mu_{i,c}$ on $\kappa$ is omitted from the notation for simplicity. An 	explicit description of $\mu_{i,c}$ can be found using the Matrix-Tree theorem on a suitable labeled digraph, see  \cite{Feliu-Simplifying}.

	\begin{example}\label{Example Finding Mu}
		Consider the following reaction  network with $\cX=\{X_1,X_2,X_3\}$ and $\cY=\{Y_1,Y_2,Y_3\}$:
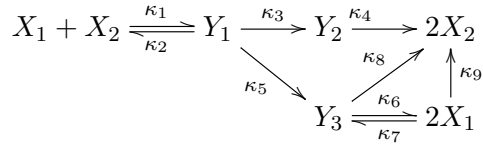
		\[\xymatrix @C=2pc @R=1.5pc{ X_1+X_2\ar@<+1pt>@{-^>}[r]^{\quad \kappa_1}\ar@<-1pt>@{_<-}[r]_{\quad \kappa_2} & Y_1\ar[r]^{\kappa_3}\ar[dr]_{\kappa_5} & Y_2\ar[r]^{\kappa_4\qquad} & 2X_2 \\
			& & Y_3\ar@<+1pt>@{-^>}[r]^{\kappa_6}\ar@<-1pt>@{_<-}[r]_{\kappa_7}\ar[ur]^{\kappa_8} & 2X_1\ar[u]_{\kappa_{9}}}\]
					The linear system in $y_1,y_2,y_3$ that the steady state polynomials of $Y_1,Y_2,Y_3$ define is:
\begin{align*}
		\kappa_1x_1x_2-(\kappa_2+\kappa_3+\kappa_5)y_1  &=  0, \\
		\kappa_3y_1-\kappa_4y_2  & =  0, \\
		\kappa_5y_1-(\kappa_6+\kappa_8)y_3+\kappa_7x_1^2  &=  0,
\end{align*}
and its solution is
\begin{align*}
	y_1 & =\tfrac{\kappa_1}{\kappa_2+\kappa_3+\kappa_5}x_1x_2, &
		y_2& =\tfrac{\kappa_1\kappa_3}{\kappa_4(\kappa_2+\kappa_3+\kappa_5)}x_1x_2, \\
		y_3& =\tfrac{\kappa_1\kappa_5}{(\kappa_6+\kappa_8)(\kappa_2+\kappa_3+\kappa_5)}x_1x_2+\tfrac{\kappa_7}{\kappa_6+\kappa_8}x_1^2.
\end{align*}
		This gives
\begin{align*}
		\mu_{1,\scriptscriptstyle X_1+X_2} & =\tfrac{\kappa_1}{\kappa_2+\kappa_3+\kappa_5}, & \mu_{1,\scriptscriptstyle 2X_1} & =0, & \mu_{1,\scriptscriptstyle 2X_2}& =0, \\
		\mu_{2,\scriptscriptstyle X_1+X_2} &=\tfrac{\kappa_1\kappa_3}{\kappa_4(\kappa_2+\kappa_3+\kappa_5)}, & \mu_{2,\scriptscriptstyle 2X_1}& =0, & \mu_{2,\scriptscriptstyle 2X_2}&=0, \\
		\mu_{3,\scriptscriptstyle X_1+X_2} & =\tfrac{\kappa_1\kappa_5}{(\kappa_6+\kappa_8)(\kappa_2+\kappa_3+\kappa_5)},& \mu_{3,\scriptscriptstyle 2X_1}& =\tfrac{\kappa_7}{\kappa_6+\kappa_8},& \mu_{3,\scriptscriptstyle 2X_2}&=0.
\end{align*}
	\end{example}
	
	\begin{example}\label{Example MAPK mu}
The following digraph corresponds to the Mitogen-Activated Protein Kinase cascade (MAPK) given in \cite{MAPK-Multistationarity}:
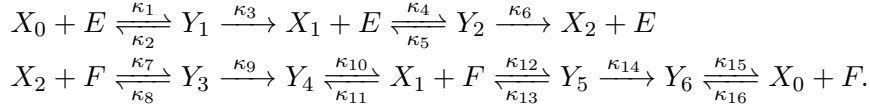
		\[\begin{array}{l}
		X_0+E\ce{<=>[\kappa_1][\kappa_2]} Y_1\ce{->[\kappa_3]} X_1+E\ce{<=>[\kappa_4][\kappa_5]} Y_2\ce{->[\kappa_6]} X_2+E\\
		X_2+F\ce{<=>[\kappa_7][\kappa_8]} Y_3\ce{->[\kappa_9]} Y_4\ce{<=>[\kappa_{10}][\kappa_{11}]} X_1+F\ce{<=>[\kappa_{12}][\kappa_{13}]} Y_5\ce{->[\kappa_{14}]} Y_6\ce{<=>[\kappa_{15}][\kappa_{16}]} X_0+F.
		\end{array}\]
				Species $Y_1,\dots,Y_6$    are intermediates.  
		The non-zero coefficients $\mu_{i,c}$ are:
		\begin{align*}
		\mu_{1,\scriptscriptstyle X_0+E}= & \tfrac{\kappa_1}{\kappa_2+\kappa_3}, &  \mu_{2,\scriptscriptstyle X_1+E}= & \tfrac{\kappa_4}{\kappa_5+\kappa_6}, &  \mu_{3,\scriptscriptstyle X_2+F}= & \tfrac{\kappa_7}{\kappa_8+\kappa_9}, & & \\
		 \mu_{4,\scriptscriptstyle X_2+F}= &\tfrac{\kappa_7\k_9}{(\kappa_8+\kappa_9)\k_{10}}, &
		\mu_{4,\scriptscriptstyle X_1+F}= & \tfrac{\kappa_{11}}{\kappa_{10}}, &
		\mu_{5,\scriptscriptstyle X_1+F}= & \tfrac{\kappa_{12}}{\kappa_{13}+\kappa_{14}},  \\
		 \mu_{6,\scriptscriptstyle X_1+F}= & \tfrac{\kappa_{12}\k_{14}}{(\kappa_{13}+\kappa_{14})\k_{15}}, & \mu_{6,\scriptscriptstyle X_0+F}= & \tfrac{\kappa_{16}}{\kappa_{15}}.
		\end{align*}
	\end{example}

\medskip
\subsection{Extended and core networks}
	
	\begin{definition}\label{Notation Extension via the intermediates}
		Let $\cN=(\cS,\cC,\cR)$ and $\widetilde{\cN}=(\widetilde{\cS},\widetilde{\cC},\widetilde{\cR})$ be two reaction networks. We say that $\widetilde{\cN}$ is an \emph{extension} of $\cN$ via the addition of   intermediates $Y_1,\dots,Y_m$ if
		\begin{enumerate}[(i)]
			\item $\cY=\{Y_1,\dots,Y_m\}$ is a set of intermediates of $\widetilde{\cN}$.
			\item $\cS\cup \cY\subseteq \widetilde{\cS}$ and  $\cC\cup \cY\subseteq \widetilde{\cC}$.
			\item $c\rightarrow c'\in\cR$ if and only if there is a directed path from $c$ to $c'$ in the  digraph associated with $\widetilde{\cN}$, such that all vertices other than $c$ and $c'$ belong to $\cY$ (there might be none).
		\end{enumerate}
		In this case $\cN$ is called the \emph{core network} of $\widetilde{\cN}$. 
	\end{definition}

	\begin{example}\label{Example Realization Condition}
		The core network associated with the network in Example \ref{Example Finding Mu} is:
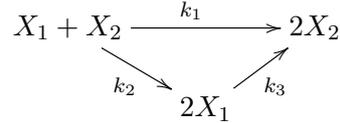
		\[\xymatrix @C=1.25pc @R=1.25pc{ X_1+X_2\ar[rr]^{k_1}\ar[dr]_{k_2} & & 2X_2 \\
			& 2X_1\ar[ru]_{k_3} & }\]
	\end{example}

		\begin{example}\label{Example MAPK mu reduced}
The core network of the network in Example \ref{Example MAPK mu} has digraph
\[
		X_0+E \ce{->[k_1]} X_1+E \ce{->[k_2]} X_2+E\qquad 
		X_2+F  \ce{->[k_3]} X_1+F \ce{->[k_4]} X_0+F.
\]
	\end{example}

	Notations $\kappa,\widetilde{I},\widetilde{F}$ are used to address reaction rate constants, steady state ideal and steady state polynomials of the extended network respectively. This notation is fixed from now on whenever    we study extensions via the addition of intermediates.

Given $\widetilde{\cN}$  an extension of $\cN$ via the addition of intermediates $Y_1,\dots,Y_m$, we define a map 
		\[		\begin{array}{lrll}
		\phi \colon & \mathbb{R}(k) & \longrightarrow & \mathbb{R}(\kappa)\\
		& k_{c\rightarrow c'} & \longmapsto & \phi_{c\rightarrow c'}(\kappa),
		\end{array}\] 
		such that for every reaction $c\rightarrow c'\in\cR$, $\phi_{c\rightarrow c'}(\kappa)$ is the rational function
\begin{equation}\label{eq:phi}
		\phi_{c\rightarrow c'}(\kappa)=  \kappa_{c\rightarrow c'}+\sum_{i=1}^m\kappa_{Y_i\rightarrow c'}\,\mu_{i,c},
		\end{equation}
		where it is understood that $\kappa_{c\rightarrow c'} =0$,  $\kappa_{Y_i\rightarrow c'}=0$ if respectively $c\rightarrow c'$,  $Y_i\rightarrow c'$ do not belong to $\widetilde{\cR}$.  Note that $\phi_{c\rightarrow c'}(\kappa)\neq 0$ for all $c\rightarrow c'$ by Definition~\ref{Notation Extension via the intermediates}(iii) and that $\phi_{c\rightarrow c'}(\kappa)$ is a rational function with positive coefficients.

 The map $\phi$ extends to a map  
		\[\Phi\colon \mathbb{R}[k][x]\rightarrow\mathbb{R}(\kappa)[y,x].\]
		For example,  if $F_i$ is a steady state polynomial of $\cN$, $\Phi(F_i)$ is the polynomial obtained by replacing $k_{c\rightarrow c'}$ by the rational function $\phi_{c\rightarrow c'}(\kappa)$.
 	If the rational functions $\phi_{c\rightarrow c'}(\kappa)$  are algebraically independent over $\R$, then
		 $\phi$ extends to a map of polynomial rings 
		\[\Phi\colon \mathbb{R}(k)[x]\rightarrow\mathbb{R}(\kappa)[y,x].\]
		We explore in Section~\ref{sec:algebraic-independent} ways to check whether the algebraic independence condition holds, and provide types of intermediates for which it holds and no extra check is required.
		
We introduce the following polynomials 
\begin{equation}\label{eq:modified-polynomials}
H_{i}(y,x)=y_i-\sum_{c\in\cC}\mu_{i,c}\,x^c \ \in \R(\k)[y,x],\qquad i=1,\dots,m.
\end{equation}
			
	\begin{thm}\label{Theorem Systems of Intermediates}\emph{(\cite[Theorems 3.1 and 3.2]{Feliu-Simplifying})}
		Let $\widetilde{\cN}$ be an extension of $\cN$ via the addition of intermediates $Y_1,\dots,Y_m$.
		\begin{enumerate}[(i)]
			\item The coefficient $\mu_{i,c}$ is nonzero if and only if the non-intermediate complex $c$ is an input for $Y_i$ in $\widetilde{\cN}$.
			\item 
			The set of steady state polynomials of non-intermediate species and the polynomials $H_1,\dots,H_m$ in \eqref{eq:modified-polynomials} form a basis of $\widetilde{I}$.
			\item 
			$\widetilde{F}_i\Big(\sum_{c\in\cC}\mu_{1,c}\,x^c,\dots,\sum_{c\in\cC}\mu_{m,c}\,x^c,x_1,\dots,x_n\Big)=\Phi(F_i(x))$ for $i=1,\dots,n$. 
		\end{enumerate}
	\end{thm}

		Statements (ii) and (iii) of the previous theorem constitute the proof of the following corollary.
		
		\begin{corollary}\label{cor:basistildeI}
	Let $B$ be the set of steady state polynomials of $\cN$. Then
		\begin{equation*}\label{equ1}
		\widetilde{I} = \Big\langle \Phi(B)\cup\{H_1(y,x),\dots,H_m(y,x)\} \Big\rangle.
		\end{equation*}
		\end{corollary}

We conclude this section
 with basic properties of $\Phi$.

 \begin{lemma}\label{lem:phi-properties}
With the notation above,  assume $\phi_{c\rightarrow c'}(\kappa)$  for all  $c\rightarrow c'\in\cR$ are algebraically independent over $\R$.   Let $B=\{f_1,\dots,f_\ell\}$ and $B'=\{f'_1,\dots,f'_{\ell'}\}$ be two sets in $\R(k)[x]$.
 \begin{enumerate}[(i)]
 \item If $f\in \langle B\rangle $, then $\Phi(f)\in\langle \Phi(B)\rangle $.  
 \item If $\langle B\rangle =  \langle B'\rangle $, then $\langle \Phi(B)\rangle =  \langle \Phi(B')\rangle$. Thus $\Phi(\langle B\rangle)$ is well defined.
 \end{enumerate}
 \end{lemma}
 \begin{proof}
 (i) Write $f= \sum_{j=1}^{\ell} \alpha_{j} f_j$ with $\alpha_j\in \R(k)[x]$. Then 
 \[\Phi(f)= \sum_{j=1}^{\ell} \Phi(\alpha_{j})\Phi(f_j)\in \langle \Phi(B)\rangle.\]

 \noindent
 (ii) It is enough to show inclusion $\subseteq$, since the other inclusion is  analogous. 
If  $g\in \langle\Phi(B)\rangle$, we have
 \[g = \sum_{i=1}^\ell \lambda_i \Phi(f_i),\qquad \lambda_i\in \R(\k)[y,x].\]
 Since $f_i\in \langle B'\rangle$, we have by (i) that   $\Phi(f_i)\in \langle \Phi(B')\rangle$.
 In particular, $g $ is an algebraic combination of the polynomials $\Phi(f_1'),\dots,\Phi(f'_{\ell'})$ with coefficients in $\R(\k)[y,x]$. Thus $g\in \langle \Phi(B')\rangle$.
 \end{proof}

	\section{Gr\"{o}bner bases and intermediates}\label{sec:groebner}

Typically, the values of the reaction rate constants are unknown and reaction networks of interest involve a considerable number of variables.
As a consequence, finding a Gr\"obner basis of the steady state ideal over the field $\R(\k)$ can be a demanding task, and sometimes even impossible with standard computers. However, the presence of intermediates, a common feature of reaction networks, can reduce the computation time substantially, by exploiting  the structure of the steady state polynomials associated with intermediates given in Theorem~\ref{Theorem Systems of Intermediates}. 	The main result of this section is Theorem~\ref{Proposition Groebner Intermediates}. Example \ref{Example Conradi System 169} illustrates how the computation time can be reduced by applying our results.

	We start with some concepts from computational algebraic geometry.

\subsection{Monomial orders and Gr\"obner bases}
We follow the notation on Gr\"obner bases from \cite{Cox}. 
We  give here a brief overview of the results required in this text.
	
Given a monomial order on $R=K[x_1,\dots,x_n]$, let $\LM(f)$ and $\LT(f)$
denote respectively the \emph{leading monomial} and \emph{leading term} of $f$.  That is,   $\LT(f)=\alpha \LM(f)$  if $\alpha$ is the coefficient of the greatest monomial of $f$. 
Then, for a subset $A\subseteq R$,  one defines 
$\LT(A)=\big\{\LT(f)\mid f\in A\big\}$ and 
$\LM(A)=\big\{\LM(f)\mid f\in A\big\}.$
Clearly, 
\begin{equation}\label{eq:monomial}
\langle\LT(A)\rangle=\langle\LM(A)\rangle.
\end{equation}

For an ideal $I$, the \emph{initial ideal} is the ideal generated by the leading terms of the elements of $I$, $\langle\LT(I)\rangle$. 
A subset $G\subseteq I$ is called a \emph{Gr\"{o}bner basis} for $I$ if 
\[\big\langle\LT(I)\big\rangle=\big\langle\LT(G)\big\rangle,\qquad (\text{equiv. }\ \big\langle\LM(I)\big\rangle=\big\langle\LM(G)\big\rangle). \]
 A Gr\"{o}bner basis is a basis of $I$ as well. Further, $G$ is a \emph{reduced Gr\"{o}bner basis} if additionally for every element $g\in G$ none of its terms can be divided by the leading monomial of an element in $G-\{g\}$, and the coefficient of $\LM(g)$ is $1$.

Whether a basis of an ideal is a Gr\"obner basis depends on the chosen monomial order.
Given an ideal and a monomial order, the Gr\"{o}bner basis is not unique but there is a unique reduced Gr\"obner basis (see \cite{Cox}).

We will use the following lemma, which follows from Lemma 2.3.1 and Theorem 2.3.2 of \cite{MonomialIdeals-HerzogHibi}.

\begin{lemma}\label{lem:relatively-prime}
Let $B$ be a basis of $I$. If the leading monomials of every pair $f,g\in B$ are relatively prime, then $B$ is a Gr\"obner basis.
\end{lemma}

   All monomial orders 
are defined via a matrix in the following way (though not all matrices $M$ define a monomial order in this way, \cite{ROBBIANO:1985jc,Cox}). For $M\in \R^{n\times n}$  with full rank, the associated order fulfills $x^{c_1}>x^{c_2}$ if the first non-zero entry of the vector $M(c_1-c_2)$ is positive 

A typical order is the  lexicographic monomial order, \emph{lex}. After choosing a variable order  $x_{a_1}>\dots>x_{a_n}$, $\lex(x_{a_1},\dots,x_{a_n})$ is the order defined by the matrix with 1 in positions $(i,a_i)$ for all $i=1,\dots,n$ and zero otherwise. 

Another monomial order of interest is the graded reverse-lexicographic order, abbreviated \emph{grevlex}. With this order, 
$x^{c_1}>x^{c_2}$ if the total degree of the first monomial is larger than the second. If they are equal, then the monomial with the smallest variable with least exponent is the greatest one. 
Grevlex with order of variables $x_1>\dots>x_n$ is defined by the matrix 
	{\small \[
	\left(\begin{array}{rrrrr}
	1 & 1 & \dots & 1 & 1\\
	0 & 0 & \dots & 0 & -1\\
	0 & 0 & \dots & -1 & 0\\
	\vdots & \vdots & \ddots & \vdots & \vdots\\
	0 & -1 & \dots & 0 & 0
	\end{array}\right).\]}

	The choice of order plays an important role in the computation time for Gr\"obner bases, performing lex typically worse than  grevlex. However, lex, as any other elimination type order, has a crucial property on elimination of variables. Given a partitioning of the set of variables, $\{x_1,\dots,x_n\}=\{x_{j_1},\dots,x_{j_{n-s}}\}\cup\{x_{i_1},\dots,x_{i_s}\}$, a monomial order is of elimination type if  $x_{j_\ell}$, for $\ell=1,\dots,n-s$, is larger than any monomial in $K[x_{i_1},\dots,x_{i_s}]$ \cite[\S 3.1, Exercise 5]{CoxElementary}. Clearly, $\lex(x_{j_1},\dots,x_{j_{n-s}}, x_{i_1},\dots,x_{i_s})$ is of elimination type. 
		If $G$ is a Gr\"{o}bner basis of $I$ with respect to an elimination type order as above,  then $G\cap K[x_{i_1},\dots,x_{i_s}]$ is a Gr\"{o}bner basis  of $I\cap K[x_{i_1},\dots,x_{i_s}]$ with respect to the induced monomial on  $K[x_{i_1},\dots,x_{i_s}]$, which  for lex is $\lex(x_{i_1},\dots,x_{i_s})$.

	\subsection{Gr\"obner bases and intermediates}
In this subsection we fix a reaction network $\cN$ and an extension $\widetilde{\cN}$  via the addition of intermediates $Y_1,\dots,Y_m$.
We show that any Gr\"obner basis of the steady state ideal of $\cN$ can be extended to one of $\widetilde{\cN}$ by simply adding the polynomials  $H_1,\dots,H_m$  given in Equation~\eqref{eq:modified-polynomials}.
By default, we order the variables $y_1>\dots>y_m>x_1>\dots>x_n$.
	We start with some general lemmas.
	
	\begin{lemma}\label{Lemma2.1}
		Let $I=\langle f_0,f_1,\dots,f_s\rangle \subseteq K[y,x_1,\dots,x_n]$ be an ideal  such that $f_i\in K[x_1,\dots,x_n]$ for $i=1,\dots,s$ and $f_0=y+f_0'$, with $f_0'\in K[x_1,\dots,x_n]$. 
		Consider a monomial order defined by a matrix $M$ whose first row is $\begin{pmatrix}
		1 & 0 & \dots & 0
		\end{pmatrix}$. Then
		 \[\big\langle\LT(I)\big\rangle=\big\langle y\big\rangle+\big\langle\LT(\langle f_1,\dots, f_s\rangle)\big\rangle.\]
		 
		 Further given $G\subseteq K[x_1,\dots,x_n]$, $G$ is a Gr\"{o}bner basis of $\langle f_1,\dots,f_s\rangle$  if and only if $\{f_0\} \cup G$ is  a Gr\"{o}bner basis of $I$.
	\end{lemma}
	
	\begin{proof}
	By the choice of monomial order, the monomial $y$ is larger than any monomial not involving $y$.
	Consider a reduced Gr\"{o}bner basis $G'$ of $\langle f_1,\dots,f_s\rangle$. Then the leading terms of the elements in $G'$ are relatively prime with each other and with the leading term of $f_0$. 
	Since $\{f_0\} \cup G'$ is a basis of $I$, then by Lemma~\ref{lem:relatively-prime} 
$\{f_0\} \cup G'$ is a Gr\"{o}bner basis of $I$. Now, the initial ideal of $I$ is generated by the leading terms of $\{f_0\} \cup G'$. So:
			\begin{align*}
		\langle\LT(I)\rangle & =\langle\LT(\{f
		_0\}\cup G')\rangle  =\langle \{y\}\cup\LT(G')\rangle  =\langle y\rangle+\langle\LT(G')\rangle\\
		& =\langle y\rangle+\langle\LT(\langle f_1,\dots,f_s\rangle)\rangle.
		\end{align*}
		This proves the first part of the lemma. 

		For the second part, 
note that 
 \[\langle y\rangle+\langle\LT(G)\rangle= \langle \{\LT(f_0)\} \cup \LT(G)\rangle=\langle\LT(\{f_0\} \cup G)\rangle .\]
Using this equality and the first part of the lemma, we have
$\{f_0\} \cup G$ is  a Gr\"{o}bner basis of $I$ if and only if 
$\langle y\rangle+\langle\LT(G)\rangle =\langle y\rangle+\langle\LT(\langle f_1,\dots,f_s\rangle)\rangle.$ 
Since $y$ is not part of any polynomial in $G$, this equality holds if and only if 
$\langle\LT(G)\rangle  = \langle\LT(\langle f_1,\dots,f_s\rangle)\rangle$, i.e. $G$ is  a Gr\"{o}bner basis of $\langle f_1,\dots,f_s\rangle$.
	\end{proof}

	Recall that we  write $I \subseteq\mathbb{R}(k)[x_1,\dots,x_n]$  and $\widetilde{I}\subseteq \mathbb{R}(\k)[y_1,\dots,y_m,x_1,\dots,x_n]$ for the steady state ideals of $\cN$ and $\widetilde{\cN}$ respectively.
	For the rest of the section, we assume that  the rational functions $\phi_{c\rightarrow c'}(\kappa)$  are \textbf{algebraically independent} over $\R$, such that $\Phi(A)$ is defined for all subsets $A$ of $\R(k)[x]$.

	 For an arbitrary basis $B$ of $I$, define
		\begin{equation}\label{Equation B tilde}
	\widetilde{B}=\Phi(B)\cup \big\{H_1(y,x),\dots,H_m(y,x) \big\}\ \subseteq\mathbb{R}(\kappa)[y,x].
	\end{equation}

	\begin{lemma}\label{Lemma Basis Intermediates}
		If $B$ is a basis of $I\subseteq\mathbb{R}(k)[x]$, then $\widetilde{B}$ is a basis of $\widetilde{I}\subseteq\mathbb{R}(\kappa)[y,x]$.
			\end{lemma}
	\begin{proof}
		Let $B'$ be the set of steady state polynomials of $\cN$. By Corollary~\ref{cor:basistildeI}
		\[\widetilde{I}=\Big\langle \Phi(B')\cup\{H_1(y,x),\dots,H_m(y,x)\} \Big\rangle.\]
		Let now $B$ be an arbitrary basis of $I$. Then $\langle B\rangle=I=\langle B'\rangle$ and thus by Lemma~\ref{lem:phi-properties}(ii),
		$\langle\Phi(B)\rangle=\langle\Phi(B')\rangle$.
		Therefore
		\[\Big\langle\Phi(B)\cup\{H_1(y,x),\dots,H_m(y,x)\} \Big\rangle = \Big\langle\Phi(B')\cup\{H_1(y,x),\dots,H_m(y,x)\}\Big\rangle=\widetilde{I}.\]
		This  completes the proof.
					\end{proof}

Let $\Rem(p,q)$ be the remainder of the division of the polynomial $p$ by $q$.

	\begin{thm}\label{Proposition Groebner Intermediates}
		Fix a monomial order on $\mathbb{R}(k)[x]$ associated with an $n\times n$ matrix $Q$, and let $G$ be a Gr\"{o}bner basis of $I$ with this order.
				 Then, $\widetilde{G}$ is a Gr\"{o}bner basis of $\widetilde{I}$ with the monomial order on $\mathbb{R}(\kappa)[y,x]$ associated with the matrix
				\begin{equation}\label{equ2}
		\widetilde{Q}=\begin{pmatrix}
		{\rm Id}_m & 0 \\
		0 & Q
		\end{pmatrix},
		\end{equation}
		where ${\rm Id}_m$ is the identity matrix of size $m$.
		
		If $G$ is reduced, then $\Phi(G)\cup \Big\{y_i-\Rem\big(\sum_{c\in\cC}\mu_{i,c}x^c,\Phi(G)\big)\Big\}$ is the reduced Gr\"{o}bner basis of $\widetilde{I}$.
	\end{thm}
	\begin{proof}  
	First note that by the monomial order given by $\widetilde{Q}$, we have $y_1>\dots>y_m>x_i$ for all $i=1,\dots,n$.  
Also, the polynomial $H_i$ has degree one in $y_i$ and  none of the elements of $\Phi(G)\cup\{H_j \mid j\neq i \}$ involves $y_i$. 

Let us assume we have shown that $\Phi(G)$ is a Gr\"obner basis of $\langle \Phi(G)\rangle$ with the given order, 
that is 
\begin{equation} \label{eq:phiGgrobner}
\big\langle\LT(\langle \Phi(G)\rangle)\big\rangle = \big\langle\LT(\Phi(G))\big\rangle.
\end{equation}
Then by Lemmas \ref{Lemma2.1} and  \ref{Lemma Basis Intermediates}, $\Phi(G) \cup \{ H_1(y,x),\dots,H_m(y,x)\}$ is a Gr\"obner  basis of $\widetilde{I}$.
Therefore the first part of the statement holds provided \eqref{eq:phiGgrobner} holds. 

Let us show \eqref{eq:phiGgrobner}. We start by noting that for a subset $J$ in $\mathbb{R}(k)[x]$, the set $\LM(J)$ consists only of monomials in $x_1,\dots,x_n$, and thus is naturally included in $\mathbb{R}(\kappa)[y,x]$ as well. 
 Further
\begin{equation} \label{eq:LMJ}
\LM(J) = \LM(\Phi(J)). 
\end{equation}

Let $G'$  be a reduced Gr\"obner basis of $I$. Since $G'$ is reduced, pairs of monomials in $\LM(G')=\LM(\Phi(G'))$ are relatively prime. Since $\Phi(G')$ is a basis of $\langle \Phi(G')\rangle$, then by Lemma~\ref{lem:relatively-prime} and Equation~\eqref{eq:monomial},  it is actually a Gr\"obner basis and \eqref{eq:phiGgrobner} holds for $G'$. 
Now, consider an arbitrary Gr\"obner basis $G$ of $I$. 
In $\mathbb{R}(k)[x]$ it holds 
\begin{equation}\label{eq:LM}
\langle \LM( G ) \rangle = \langle \LM( G' ) \rangle. 
\end{equation}
This means that every monomial in $ \langle \LM( G' ) \rangle $ is divisible by a monomial in $\langle \LM( G ) \rangle$ and viceversa \cite[\S 2.4, Lemma 2]{Cox}. Since this fact holds also in $\mathbb{R}(\kappa)[y,x]$, \eqref{eq:LM} holds also   in $\mathbb{R}(\kappa)[y,x]$. Combined with \eqref{eq:LMJ} this gives
\[ \big\langle \LM( \Phi(G) ) \big\rangle =\big \langle \LM( \Phi(G') )\big\rangle. \]

By Lemma~\ref{lem:phi-properties}(ii), $\langle G \rangle=\langle G' \rangle$  in $\mathbb{R}(k)[x]$  implies  $\langle \Phi(G) \rangle=\langle \Phi(G')  \rangle$. Thus in $\mathbb{R}(\kappa)[y,x]$ we have
\begin{align*}
\big\langle \LM( \Phi(G) )\big \rangle =\big  \langle \LM( \Phi(G') )\big \rangle = \big\langle \LM(  \langle\Phi(G') \rangle  ) \big\rangle = \big\langle \LM(  \langle\Phi(G) \rangle  ) \big\rangle.
\end{align*}
This shows that \eqref{eq:phiGgrobner} holds.

		The second part of the lemma is clear from the definition of a reduced Gr\"{o}bner basis and using that $\Phi(G)\cup\{y_i-\Rem\big(\sum_{c\in\cC}\mu_{i,c}x^c,\Phi(G)\big)\}$ is also a Gr\"{o}bner basis.
	\end{proof}
	
	From the  computational point of view, Theorem~\ref{Proposition Groebner Intermediates} is very useful. 
Instead of computing a Gr\"{o}bner basis of $\widetilde{I}$ directly, one can first compute a Gr\"{o}bner basis $G$ for the core network, with a smaller number of variables and polynomials, then add the polynomials $y_i-\sum_{c\in\cC}\mu_{i,c}x^c$, and, finally,  simplify them using polynomial division by $\Phi(G)$. The second step involves only  linear algebra. 
A possible issue here is to verify that  the rational functions $\phi_{c\rightarrow c'}$ are algebraically independent. We provide in  Section~\ref{sec:algebraic-independent} a list of network structures involving intermediates for which the condition is fulfilled.

	\begin{figure}[t]
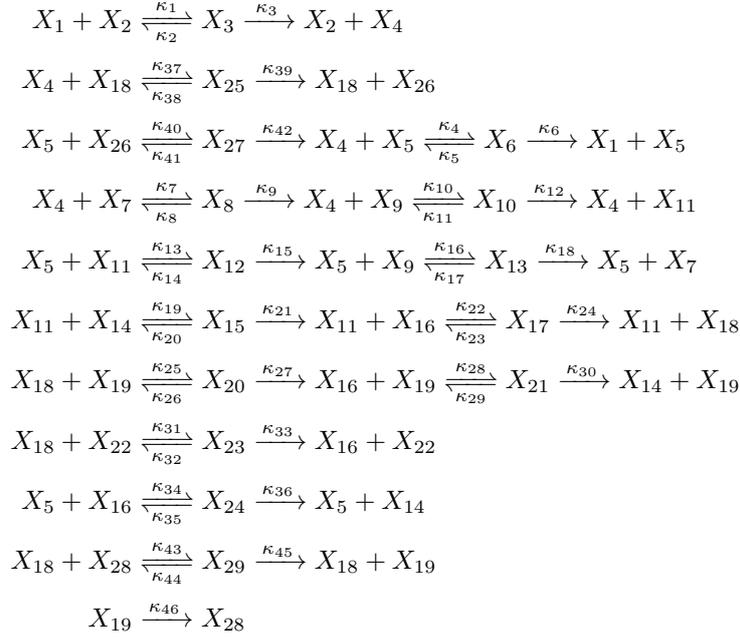

	{\small 				\begin{align*}
		X_1 + X_2  & \ce{<=>[\kappa_1][\kappa_2]} X_3 \ce{->[\kappa_3]} X_2 + X_4 \nonumber \\
X_4 + X_{18} & 		\ce{ <=>[\kappa_{37}][\kappa_{38}]} X_{25} \ce{->[\kappa_{39}]} X_{18} +X_{26}  \nonumber  \\
X_{5} + X_{26}&		\ce{  <=>[\kappa_{40}][\kappa_{41}]} X_{27} \ce{->[\kappa_{42}]} X_4 + X_5 \ce{<=>[\kappa_4][\kappa_5]} X_6 \ce{->[\kappa_6]} X_1 + X_5  \nonumber  \\
X_4 + X_7 & 		\ce{ <=>[\kappa_7][\kappa_8]} X_8 \ce{->[\kappa_9]} X_4 + X_9 \ce{<=>[\kappa_{10}][\kappa_{11}]} X_{10} \ce{->[\kappa_{12}]} X_{4} + X_{11}  \nonumber  \\
X_{5} + X_{11} & 		\ce{ <=>[\kappa_{13}][\kappa_{14}]} X_{12} \ce{->[\kappa_{15}]} X_5 + X_9 \ce{<=>[\kappa_{16}][\kappa_{17}]} X_{13} \ce{->[\kappa_{18}]} X_5 + X_7  \nonumber  \\
X_{11} + X_{14} &		\ce{ <=>[\kappa_{19}][\kappa_{20}]} X_{15} \ce{->[\kappa_{21}]} X_{11} + X_{16} \ce{<=>[\kappa_{22}][\kappa_{23}]} X_{17} \ce{->[\kappa_{24}]} X_{11} + X_{18} \label{Conradi Reaction Netwrok} \\
X_{18} + X_{19} &		\ce{<=>[\kappa_{25}][\kappa_{26}]} X_{20} \ce{->[\kappa_{27}]} X_{16} + X_{19} \ce{<=>[\kappa_{28}][\kappa_{29}]} X_{21} \ce{->[\kappa_{30}]} X_{14} + X_{19}  \nonumber  \\
X_{18} + X_{22} &		\ce{<=>[\kappa_{31}][\kappa_{32}]} X_{23} \ce{->[\kappa_{33}]} X_{16} + X_{22}  \nonumber  \\
X_{5} + X_{16} &		\ce{<=>[\kappa_{34}][\kappa_{35}]} X_{24} \ce{->[\kappa_{36}]} X_{5} + X_{14}  \nonumber  \\
X_{18} + X_{28} &		\ce{<=>[\kappa_{43}][\kappa_{44}]} X_{29} \ce{->[\kappa_{45}]} X_{18} + X_{19}   \nonumber  \\
X_{19} & 		\ce{->[\kappa_{46}]} X_{28}  \nonumber 
		\end{align*}}
	\caption{Reaction network of Example \ref{Example Conradi System 169}.} \label{Conradi Reaction Netwrok}
	\end{figure}

	\begin{example}\label{Example Conradi System 169}
An interesting example to show the advantage of using Theorem~\ref{Proposition Groebner Intermediates} is Example 4.4 of \cite{Conradi-Binomial}. We consider the reaction network $\widetilde{\cN}$ with associated digraph given in Figure~\ref{Conradi Reaction Netwrok}.

		This reaction network has 29 species and 46 reactions. Therefore the steady state ideal is generated by 29 polynomials in 29 variables and 46 parameters. Using Singular \cite{DGPS} and monomial order grevlex with $x_1>\dots>x_{29}$ (the same monomial order that is used in \cite{Conradi-Binomial}), it took between 110 and 115 seconds\footnote{Information about the processor: Intel(R) Core(TM) i5-3570 CPU @3.4GHz 3.4GHz with 8GB RAM. We report the interval of obtained times after several runs of Singular, computed in milliseconds.} to compute the reduced Gr\"{o}bner basis. This basis consists of 169 polynomials.
		
		Now we consider the monomial order introduced in Theorem~\ref{Proposition Groebner Intermediates} for the removal of the 15 intermediates: \[X_3,X_6,X_8,X_{10},X_{12},X_{13},X_{15},X_{17},X_{20},X_{21},X_{23},X_{24},X_{25},X_{27},X_{29}.\]
		The original network $\widetilde{\cN}$ is an extension of the following core network $\cN$ with 14 species and 16 reactions:
	{\small 	\begin{align*}
X_1 + X_2 &	\ce{->[k_1]} X_2 + X_4 & 
X_{5} + X_{26} &	\ce{ ->[k_3]} X_4 + X_5 \ce{->[k_4]} X_1 + X_5  \\
X_4 + X_{18}  &	\ce{->[k_2]} X_{18} + X_{26}   & 
  X_4 + X_7  &	\ce{->[k_5]} X_4 + X_9 \ce{->[k_6]} X_{4} + X_{11}\nonumber \\
  X_{18} + X_{22} &	\ce{ ->[k_{13}]} X_{16} + X_{22} & X_{5} + X_{11} &	\ce{ ->[k_7]} X_5 + X_9 \ce{->[k_8]} X_5 + X_7  \\ 
  X_{5} + X_{16}  &	\ce{->[k_{14}]} X_{5} + X_{14} &  
X_{11} + X_{14}  &	\ce{->[k_9]} X_{11} + X_{16} \ce{->[k_{10}]} X_{11} + X_{18} \\ 
X_{18} + X_{28} &	\ce{ ->[k_{15}]} X_{18} + X_{19} & X_{18} + X_{19}  &	\ce{->[k_{11}]} X_{16} + X_{19} \ce{->[k_{12}]} X_{14} + X_{19} 
 \nonumber \\
X_{19}  &	\ce{->[k_{16}]} X_{28}. \nonumber
		\end{align*}}
The functions $\phi_{c\rightarrow c'}$ are algebraically independent over $\R$ by Corollary \ref{Corollary Algebraic Independence singletons} in Section~\ref{sec:algebraic-independent}.
We consider grevlex with $x_1>\dots>x_{28}$ for the monomials corresponding to $\cN$.
The monomial order  in Theorem~\ref{Proposition Groebner Intermediates} 	
is then associated with the following matrix
		{\small \begin{equation*}
		\widetilde{Q}=\left(\begin{array}{c;{2pt/2pt}c}
		\begin{array}{ccc}
		1 & & 0\\
		& \ddots & \\
		0 & & 1 
		\end{array} & 0\\
		\hdashline[2pt/2pt]
		0 & \begin{array}{rrrrr}
		1 & 1 & \dots & 1 & 1\\
		0 & 0 & \dots & 0 & -1\\
		0 & 0 & \dots & -1 & 0\\
		\vdots & \vdots & & \vdots & \vdots\\
		0 & -1 & \dots & 0 & 0
		\end{array}
		\end{array}\right),
		\end{equation*}}
		and order of variables
		{\small \[\begin{array}{l}
		x_3>x_6>x_8>x_{10}>x_{12}>x_{13}>x_{15}>x_{17}>x_{20}>x_{21}>x_{23}>x_{24}\\
		\qquad >x_{25}>x_{27}>x_{29}>x_1>x_2>x_4>x_5>x_7>x_9>x_{11}>x_{14}>x_{16}\\
		\qquad > x_{18}>x_{19}>x_{22}>x_{26}>x_{28}.
		\end{array}\]}
The reduced Gr\"obner basis of $\widetilde{I}$ with this monomial order has 33 polynomials and it takes about 96 seconds to compute it directly with Singular. Alternatively the strategy outlined in 
  Theorem~\ref{Proposition Groebner Intermediates} can be applied. The steady state ideal $I$ of $\cN$ is generated by 11 polynomials in 14 variables and 16 parameters. 
Using Singular, the reduced Gr\"{o}bner basis of  $I$ has 18 polynomials and its computation takes less than a millisecond.  The computation time for  the polynomials $H_i(y,x)$ is neglectable, since they are found by solving 15 independent linear equations. Therefore the reduced Gr\"{o}bner basis of the ideal of the original system has 18+15=33 polynomials and can be computed in less than a millisecond.
 
 We conclude that in general, regarding computational time, the monomial order introduced in Theorem~\ref{Proposition Groebner Intermediates} is a good choice for networks with intermediates, and further,  by applying the strategy of Theorem~\ref{Proposition Groebner Intermediates} we reduce the computation time considerably, compared with direct computation of the reduced Gr\"obner basis.
	\end{example}
	
	\begin{remark}
 Theorem~\ref{Proposition Groebner Intermediates}  holds regardless the choice of method to compute a Gr\"{o}bner basis. Since the computation of the polynomials $H_i$ is simple linear algebra, even for the fastest available methods for the computation of Gr\"obner bases, decomposing the computation as in Theorem~\ref{Proposition Groebner Intermediates}  should be faster than direct computation of the basis of the steady state ideal of $\widetilde{\cN}$.
	\end{remark}

	\begin{remark}
	For polynomials with integer coefficients, it is usually faster to compute a Gr\"obner basis using the so-called  \emph{p-modular} approach, see e.g. \cite{p-modular-Winkler, p-modular-Noro-Yokoyama}. These methods first choose a so-called lucky prime and compute a Gr\"{o}bner basis of the ideal in $\overline{\mathbb{Z}}_p[x]$. Then the coefficients of this Gr\"{o}bner basis are lifted to a Gr\"{o}bner basis in $\mathbb{Q}[x]$.  For the sake of comparison, we also computed how long it takes to find a Gr\"obner basis using p-modular approaches on the extended network in  Example \ref{Example Conradi System 169} with grevlex and $x_1>\dots>x_{29}$. 
	Using the largest prime number in Singular, $p= 32003$, it takes 127 seconds to compute the Gr\"{o}bner basis over $\overline{\mathbb{Z}}_{32003}$. Since coefficients in the starting basis are $1$ or $-1$, one may think that $p=2$ is a lucky prime. It took 97 seconds to compute the Gr\"{o}bner basis over $\overline{\mathbb{Z}}_2$. These times are larger than the times reported in Example~\ref{Example Conradi System 169}
 (and these Gr\"obner bases still need to be lifted to 	$\mathbb{Q}(\kappa)[y,x]$).  
	\end{remark}

An important consequence of Theorem~\ref{Proposition Groebner Intermediates} concerns parameter-free model discrimination. In this setting one seeks elements of the steady state ideal $\widetilde{I}$ involving only the concentrations of species that are experimentally measurable. These elements are called \emph{invariants}. Each invariant implies that there is a set of monomials that lie on a hyperplane, and the hypothesis of coplanary is then tested using experimental data \cite{Wnt-Matroid,InvariantModelDiscrimination,ComplexLinearInvariantsGunawardeena,
GeometryOfMultisitePhosphorylationGunawardeena,DistributiveProcessiveMultisitePhosphorylationGunawardeena}. This approach is attractive because it does not require  knowing the values of the reaction rate constants.

Experimentally measurable species do not typically involve intermediates. In this case, Theorem~\ref{Proposition Groebner Intermediates} tells us that invariants on the non-intermediate species can be computed directly from the core network, using elimination ideals.

\begin{corollary}\label{Corollary Elimination Intermediates}
		Let $\cN$ be a reaction network and $\widetilde{\cN}$ an extension of it via the addition of $m$ intermediates $Y_1,\dots,Y_m$. Let $X_{i_1},\dots,X_{i_p}$ be non-intermediates. Then 
		\[\widetilde{I}\cap\mathbb{R}(\kappa)[x_{i_1},\dots,x_{i_p} ]= \Phi(I  \cap \mathbb{R}(k)[x_{i_1},\dots,x_{i_p}]).\]
	\end{corollary}
	\begin{proof}
For simplicity, assume $\{i_1,\dots,i_p\}= \{n-p+1,\dots,n\}$ and let $\overline{x} = (x_{n-p+1},\dots,x_n)$.
		Consider the monomial order $\lex(y_1,\dots,y_m,x_1,\dots,x_n)$  on $\mathbb{R}(\kappa)[y,x]$, and $\lex(x_1,\dots,x_n)$ on $\mathbb{R}(k)[x]$. Let $G$ be a Gr\"{o}bner basis of $I$. By Theorem \ref{Proposition Groebner Intermediates}, $\widetilde{G}$ is a Gr\"{o}bner basis of $\widetilde{I}$. By the properties of lex and  Lemma~\ref{lem:phi-properties}(ii) we have
		\[\widetilde{I}\cap\mathbb{R}(\kappa)[\overline{x}] = \langle \widetilde{G} \cap\mathbb{R}(\kappa)[\overline{x}]  \rangle = \langle \Phi( G \cap  \mathbb{R}(k)[\overline{x}] ) \rangle =  \Phi( I  \cap \mathbb{R}(k)[\overline{x}]) .\]
This concludes the proof.
	\end{proof}
	
Note that the monomial order on $\R(\k)[y,x]$ given in Theorem~\ref{Proposition Groebner Intermediates} is of elimination type with respect to the partition $\{y_1,\dots,y_m\}\cup \{x_1,\dots,x_n\}$.

	\begin{example}\label{reduced MAPK}
	Consider the network in Example \ref{Example MAPK mu} and its core network in Example \ref{Example MAPK mu reduced}. In order to find invariants of the extended network involving the concentration of the non-intermediate species $E,X_0,X_1,X_2$, we consider the ideal $I\cap \R(k)[e,x_0,x_1,x_2]$, which is generated by the polynomial
	\[e\, ( k_1 k_3  x_0x_2 -k_2k_4x_1^2).\]
We have 
\begin{align*}
\phi(k_1,k_2,k_3,k_4)  & = \left(  \tfrac{\kappa_1\k_3}{\kappa_2+\kappa_3}  , \tfrac{\kappa_4\k_6}{\kappa_5+\kappa_6},   \tfrac{\kappa_7\k_9}{\kappa_8+\kappa_9}, \tfrac{\kappa_{12}\k_{14}}{\kappa_{13}+\kappa_{14}}\right).
\end{align*} 
The functions $\phi_{c\rightarrow c'}$ are algebraically independent over $\R$ by Corollary \ref{Corollary Algebraic Independence singletons}.
By Corollary \ref{Corollary Elimination Intermediates} the ideal $\widetilde{I}\cap \R(\k)[e,x_0,x_1,x_2]$ is
generated by the polynomial 
	\[e \left(  \tfrac{\kappa_1\k_3}{\kappa_2+\kappa_3}     \tfrac{\kappa_7\k_9}{\kappa_8+\kappa_9} x_0x_2   -   \tfrac{\kappa_4\k_6}{\kappa_5+\kappa_6} \tfrac{\kappa_{12}\k_{14}}{\kappa_{13}+\kappa_{14}} x_1^2\right).\]
		\end{example}
	
	\subsection{Detecting binomial steady state ideals} A \emph{binomial} is a polynomial having at most two terms. 
		An ideal is said to be binomial if it admits a set of generators consisting of binomials only.  By \cite[Corollary 1.2]{Eisenbud-Binomial}, an ideal is binomial if and only if any reduced Gr\"obner basis (with respect to any monomial order)  consists of binomials. 
		
	It is of biological relevance  in the study of reaction networks   to determine whether there exists a choice of reaction rate constants $k$ for which there are multiple positive steady states in some coset $x_0+S$  defined by the vector subspace $S$ that contains the image of $F_{k}$ (see Section~\ref{sec:steady_state_ideal}). This property is termed \emph{multistationarity}. 
If the steady state ideal is binomial, then there exist efficient ways to determine whether the network admits multistationarity  \cite{Dickenstein-Toric,Dickenstein-MESSI,Feliu-Sign}. This leads to the problem of determining whether an ideal is binomial, and in case it is, of finding a binomial basis of it. 
As noted, both questions can be addressed by finding a Gr\"obner basis of the steady state ideal of the network. Thus, for networks with intermediates, our results can be applied also to detect binomial steady state ideals.

Recall that we are assuming that the rational functions $\phi_{c\rightarrow c'}(\kappa)$  are \textbf{algebraically independent} over $\R$.

	\begin{thm}\label{Theorem Binomiality and Intermediates}
		Let $\cN$ be a reaction network and $\widetilde{\cN}$ an extension of it via the addition of $m$ intermediates $Y_1,\dots,Y_m$. 
		
		The steady state 	$\widetilde{I}$ is  binomial if and only if 
		\begin{itemize}
		\item 		$I$ is binomial, and,
		\item  for any reduced Gr\"obner basis $G$ of $I$ and for every $i=1,\dots,m$, the remainder of the division of  $\sum_{c\in \cC}\mu_{i,c}x^c$ by $\Phi(G)$ has at most one term.
		\end{itemize}
	\end{thm}
	\begin{proof}
	Fix any monomial order on $\R(k)[x_1,\dots,x_n]$ associated with an $n\times n$ matrix $Q$ and consider the monomial order with matrix $\widetilde{Q}$  from Theorem~\ref{Proposition Groebner Intermediates}. 
Let $G$ be the reduced Gr\"obner basis of $I$ and
\[
\widetilde{G}' = \Phi(G)\cup \Big\{y_i-\Rem\big(\sum_{c\in\cC}\mu_{i,c}x^c,\Phi(G)\big)\Big\}
\]
	 the reduced Gr\"{o}bner basis of $\widetilde{I}$ (cf. Theorem~\ref{Proposition Groebner Intermediates}).
Using that an ideal is binomial if and only if any reduced Gr\"obner basis consists of binomials, the theorem is a consequence of the following two facts:
	\begin{itemize}
		\item By definition, $\widetilde{G}'$ consists of binomials if and only if $\Phi(G)$ is a set of binomials and the remainder of the division of  $\sum_{c\in \cC}\mu_{i,c}x^c$ by $\Phi(G)$ has at most one term. 
		\item By the algebraic independence of $\phi_{c\rightarrow c'}$, $\Phi(G)$ consists of binomials if and only if $G$ does. 
		\end{itemize}
	\end{proof}

Since the polynomial $\sum_{c\in \cC}\mu_{i,c}x^c$ has exactly one term  for 1-input intermediates, we readily obtain the following corollary.

	\begin{corollary}\label{Corollary on-input and binomiality}
	Let $\cN$ be a reaction network and $\widetilde{\cN}$ an extension of it via the addition of $m$ \emph{1-input} intermediates $Y_1,\dots,Y_m$.  Then $\widetilde{I}$ is binomial if and only if $I$ is binomial.
	\end{corollary}

Since  1-input intermediates are the most abundant form of intermediates found in realistic networks, this corollary implies that in order to check whether a steady state ideal is binomial, we can often remove intermediates and check whether the steady state ideal of the core network is binomial.

	\begin{example}\label{Example Binomiality of Example 3.8}
Consider the network in Example \ref{Example Finding Mu} and its core network given in Example \ref{Example Realization Condition}. The functions $\phi_{c\rightarrow c'}$ are algebraically independent over $\R$ by Example \ref{Example Algebraic Independence}. 
Since the steady state ideal of $\cN$ is
\[ \langle -(k_1-k_2)x_1x_2-2k_3x_1^2, (k_1-k_2)x_1x_2+2k_3x_1^2 \rangle, \]  the core network has a binomial steady state ideal. The reduced Gr\"obner basis for this ideal with monomial order $\lex(x_1,x_2,x_3)$ is 
\[G=\left\{x_1^2-\tfrac{(k_1-k_2)}{2k_3}x_1x_2\right\}.\] 
We apply Theorem \ref{Theorem Binomiality and Intermediates} to conclude that the steady state ideal of the extended network is also binomial. 
The intermediates $Y_1,Y_2$ are 1-input intermediates and hence the remainder condition of the theorem is automatically fulfilled. For the intermediate $Y_3$, $\Rem\big(\mu_{3,\scriptscriptstyle X_1+X_2}x_1x_2+\mu_{3,\scriptscriptstyle X_2+X_3}x_2x_3,\Phi(G)\big)$ has a single term with monomial $x_1x_2$. Therefore we conclude that the extended network also has a binomial steady state ideal.
	\end{example}

			The following example shows that extended networks with multi-input intermediates might not have binomial steady state ideals, even though their core networks have.

		\begin{example}\label{Example MAPK Trinomiality}
		Consider the network given in Example  \ref{Example MAPK mu}  and its core network given in Example \ref{Example MAPK mu reduced}.		The steady state ideal of the core network is binomial with basis 
$ B=\{k_1x_0e-k_4x_1f,k_2x_1e-k_3x_2f\}.$
The intermediates $Y_4$ and $Y_6$ are 2-input intermediates. 		
		 The remainder of the division of $\mu_{4,\scriptscriptstyle X_2+E}x_2f+\mu_{4,\scriptscriptstyle X_1+F}x_1f$ by $\Phi(G)$ for $G$ the reduced Gr\"{o}bner basis of $I$  with the monomial order $\lex(x_2,x_1,x_0,f,e)$ is
		 \[\tfrac{\k_{11}}{\kappa_{10}}x_1f+\tfrac{\kappa_7\kappa_9}{\kappa_8\kappa_{10}+\kappa_9\kappa_{10}}x_2f,\] which has two terms. Therefore by  Theorem \ref{Theorem Binomiality and Intermediates} the steady state ideal of the network in Example  \ref{Example MAPK mu} is not binomial.
	\end{example}
 
\begin{remark}
In \cite{Conradi-Binomial}, a method for determining whether a \emph{homogeneous} ideal is binomial is introduced. 	 The method avoids the computation of Gr\"obner bases and is regarded as a fast method. 
If the steady state ideal of the core network is homogeneous, then Theorem~\ref{Theorem Binomiality and Intermediates} or Corollary~\ref{Corollary on-input and binomiality} in combination with this method provide a fast procedure to detect binomial steady state ideals.

 Interestingly,
steady state polynomials of core networks are often homogeneous of degree two, since  it is common that non-intermediate species appear in complexes of the form $X_i+X_j$, yielding quadratic terms in the steady state polynomials. This is for example the case for so-called Post-Translational Modification Networks \cite{TG-rational}.   
\end{remark}

	\section{Algebraic independence} \label{sec:algebraic-independent}

In this section we discuss how to check whether the functions $\phi_{c\rightarrow c'}$ are algebraically independent over $\R$ and provide classes of intermediates for which this property holds.
Consider a set of rational functions $A=\big\{\frac{f_1}{g_1},\dots,\frac{f_m}{g_m}\big\}\subseteq \R(x_1,\dots,x_n)$.
By \S III.7,  Theorem III, in  \cite{AlgebraicGeometryMethods}, the set $A$ is algebraically independent over $\R$ if and only if the rank of the associated Jacobian matrix $\begin{pmatrix}
\tfrac{\partial (f_i/g_i)}{\partial x_j}
\end{pmatrix}_{i,j}$ over $\R(x)$  is  $m$.

Another way to check algebraic independence that requires the computation of a Gr\"obner basis is as follows. Let $\varphi$ be the function on $\R^n$ minus the zero locus of the product $g_1\cdots g_m$ defined by
\[x= (x_1,\dots,x_n) \mapsto \left(\frac{f_1(x)}{g_1(x)},\dots,\frac{f_m(x)}{g_m(x)} \right).\]
By \S 3.3,  Theorem 2, in \cite{CoxElementary}, the closure of $\im(\varphi)$ is the variety associated with the ideal
\[ J:=\big\langle g_1 T_1 - f_1,\dots, g_m T_m - f_m, 1- y g_1\cdots g_m\big\rangle \cap \R[T_1,\dots,T_m]. \]
Since the sets of polynomials vanishing on a set and on its closure agree (see \cite{CoxElementary} after Definition 2 in \S 4.4),  $A$ is algebraically independent over $\R$ if and only if $J=\{0\}$.

\begin{example}\label{Example Algebraic Independence}  
	The functions $\phi_{c\rightarrow c'}$ of Examples \ref{Example Finding Mu} and \ref{Example Realization Condition} are
	\begin{align*}
	\phi_{\scriptscriptstyle X_1+X_2\rightarrow 2X_2}(\kappa) &= \kappa_4\, \mu_{2,\scriptscriptstyle X_1+X_2}+\kappa_8\, \mu_{3,\scriptscriptstyle X_1+X_2} =  
	 \tfrac{\kappa_1\kappa_3}{\kappa_2+\kappa_3+\kappa_5} + \tfrac{\kappa_1\kappa_5\kappa_8}{(\kappa_6+\kappa_8)(\kappa_2+\kappa_3+\kappa_5)} ,\\
	\phi_{\scriptscriptstyle X_1+X_2\rightarrow 2X_1}(\kappa) &= \kappa_6\, \mu_{3,\scriptscriptstyle X_1+X_2} =  \tfrac{\kappa_1\kappa_5\kappa_6}{(\kappa_6+\kappa_8)(\kappa_2+\kappa_3+\kappa_5)},\\
	\phi_{\scriptscriptstyle 2X_1\rightarrow 2X_2}(\kappa) &= \kappa_{9}+\kappa_8\, \mu_{3,\scriptscriptstyle 2X_1}= \kappa_{9}+\tfrac{\kappa_7\kappa_8}{\kappa_6+\kappa_8}.
	\end{align*}
	We find that  $J=\{0\}$. Hence the algebraic independence condition holds for the network in Example \ref{Example Realization Condition}.
	Alternatively, one easily checks that the associated Jacobian matrix has rank 3.
\end{example}

\medskip
The computations above can be simplified by taking into account what parameters occur in each of the rational functions. 

\begin{definition}\label{Definition Algebraic Independence Components}
	Let $\widetilde{\cN}$ be an extension of $\cN$ via the addition of the intermediates $\{Y_1,\dots,Y_m\}$. 
	Consider the digraph associated with  $\widetilde{\cN}$. Let $\mathcal{Y}_1,\dots,\mathcal{Y}_{t'}$ denote the vertex sets of the connected components of the subgraph induced by the subset of vertices  $\{Y_1,\dots,Y_m\}$.
	
	Let $\cR'\subseteq \cR$ be the subset of reactions of the core network that are not in $\widetilde{\cR}$. These reactions arise necessarily from paths through intermediates.
	We say that two reactions $r_1\colon c_1\rightarrow c_1', r_2\colon c_2\rightarrow c_2'\in \cR'$ \emph{overlap} if there 
	exist 	paths through intermediates 
	\[ c_1\rightarrow Y_{i_1}\rightarrow \dots \rightarrow Y_{i_p} \rightarrow  c_1', \qquad  c_2\rightarrow Y_{j_1} \rightarrow \dots \rightarrow Y_{j_q} \rightarrow  c_2' \]
	with all intermediates belonging to the same set $\mathcal{Y}_i$.
	
	Consider the equivalence relation on $\cR'$ generated by the overlap relation: $r\sim r'$ if and only if there exist $r_0=r,r_1,\dots,r_{p}=r'$ such that 
	$r_i,r_{i+1}$ overlap for all $i=0,\dots,p-1$.
	Let  $\cR'_1,\dots,\cR'_t$ be the equivalence classes of this equivalence relation.
 
\end{definition}

\begin{example}\label{Example Algebraic Independence Graphs1}
	Consider the network in Example \ref{Example Realization Condition}. The set $\cR'$  consists of two reactions
	$X_1+X_2\rightarrow 2X_2$ and $X_1+X_2\rightarrow 2X_1$. 
	The subgraph of the digraph associated with $\widetilde{\cN}$ induced by the set of intermediates is connected. Thus the two reactions of $\cR'$ are equivalent and there is one equivalence class.
\end{example}

\begin{lemma}\label{Lemma Algebraic Independece Components}
	The set $\{\phi_{c\rightarrow c'}(\kappa)\mid c\rightarrow c'\in\cR\}$ is algebraically independent over $\R$ if and only if the set $\{\phi_{c\rightarrow c'}(\kappa)\mid c\rightarrow c'\in\cR_i'\}$ is algebraically independent over $\R$ for all $i=1,\dots,t$.
\end{lemma}
\begin{proof}
Since $\cR_i'\subseteq \cR$ for all $i=1,\dots,t$, the forward implication is clear.
	
	To prove the reverse implication, assume that the sets $T_i=\{\phi_{c\rightarrow c'}(\kappa)\mid c\rightarrow c'\in\cR_i'\}$ are algebraically independent over $\R$  for all $i=1,\dots,t$.  By construction, the sets of parameters appearing in the rational functions $\phi_{c\rightarrow c'}(\k)$ are disjoint for two reactions in different  equivalence classes. Therefore the union of the sets $T_1,\dots,T_t$ is algebraically independent over  $\R$. 
 Furthermore if $c\rightarrow c'\in \cR \setminus \cR'$, then the parameter $\kappa_{c\rightarrow c'}$ appears only  in $\phi_{c\rightarrow c'}(\kappa)$. As a consequence  the set 
  \[ \bigcup_{i=1}^t T_i \cup  \{\phi_{c\rightarrow c'}(\kappa)\mid c\rightarrow c'\in\cR\setminus \cR'\} = \{\phi_{c\rightarrow c'}(\kappa)\mid c\rightarrow c'\in\cR\} \]
 is algebraically independent over $\R$. 
\end{proof}

\begin{example}\label{Example Algebraic Independence Graphs}
	Consider the network in Example \ref{Example Algebraic Independence Graphs1}. The algebraic independence of the functions $\phi_{c\rightarrow c'}(\k)$ for all reactions $c\rightarrow c'$ in $\cR$ 
	follows in this case from the algebraic independence of the functions $\phi_{c\rightarrow c'}(\k)$ for the reactions
		$X_1+X_2\rightarrow 2X_2$ and $X_1+X_2\rightarrow 2X_1$.
\end{example}

\begin{corollary}\label{Corollary Algebraic Independence singletons}
If   $\cR'=\emptyset$ or each of the equivalence classes $\cR'_1,\dots,\cR'_t$ consist of one reaction, then the rational functions $\phi_{c\rightarrow c'}(\k)$ are algebraically independent over $\R$.
\end{corollary}

For the networks in Example \ref{Example MAPK mu} and Example \ref{Example Conradi System 169}, each of the equivalence classes consist of one reaction. Therefore, by Corollary \ref{Corollary Algebraic Independence singletons}, the algebraic independence condition holds.

We  next show that the algebraic independence condition holds for specific classes of intermediates without the need of doing any extra computation. 

\begin{lemma}\label{Lemma Algebraic Independence more classes}
For the following extension networks, with intermediates $Y_1,\dots,Y_m$, the set $\{\phi_{c\rightarrow c'}(\kappa) \mid c\rightarrow c'\in \cR\}$ is algebraically independent over $\R$. 

	\begin{enumerate}[(i)]
		\item $c\ce{<->}Y_1 \ce{<->} Y_2\ce{<->} \dots \ce{<->} Y_m \ce{<->} c'$, provided $\{Y_1,\dots,Y_m\}$ is  a set of intermediates and where $\ce{<->}$ means the reaction can be irreversible or reversible.

		\item
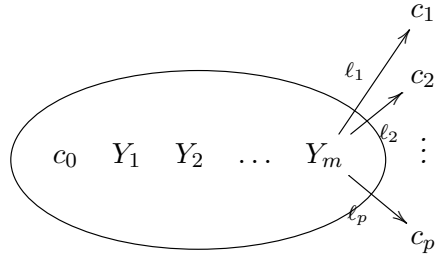
		\[
		\xymatrix @C=0.5pc @R=1pc{
			& & & & & & & c_1\\
			& & & & & & & c_2\\
			& c_0 & Y_1 & Y_2 & \dots & Y_m\ar[uurr]^{\ell_1}\ar[urr]_{\ell_2}\ar[drr]_{\ell_p} & & \vdots\\
			& & & & & & & c_p
			\save "2,1"."4,7"*\frm{e}
			\restore
		}\]
			with an arbitrary digraph structure among the complexes $c_0,Y_1,\dots,Y_m$ such that there exists a directed path from $c_0$ to $Y_m$. 
			
		\item
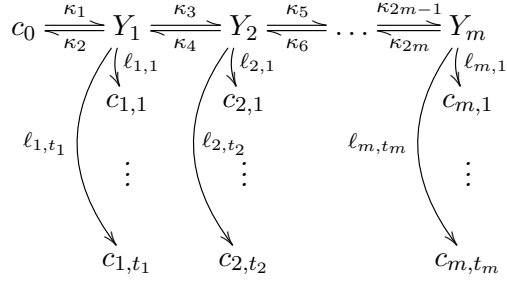
		\[
		\xymatrix@C=1.5pc @R=1pc{
			c_0\ar@<+1pt>@{-^>}[r]^{\kappa_1}\ar@<-1pt>@{_<-}[r]_{\kappa_2} & Y_1\ar@<+1pt>@{-^>}[r]^{\kappa_3}\ar@<-1pt>@{_<-}[r]_{\kappa_4}\ar@/_1ex/[d]^{\ell_{1,1}}\ar@/_4ex/[ddd]_{\ell_{1,t_1}} & Y_2\ar@<+1pt>@{-^>}[r]^{\kappa_5}\ar@<-1pt>@{_<-}[r]_{\kappa_6}\ar@/_1ex/[d]^{\ell_{2,1}}\ar@/_4ex/[ddd]^{\ell_{2,t_2}} & \dots\ar@<+1pt>@{-^>}[r]^{\kappa_{2m-1}}\ar@<-1pt>@{_<-}[r]_{\kappa_{2m}} & Y_m \ar@/_1ex/[d]^{\ell_{m,1}}\ar@/_4ex/[ddd]_{\ell_{m,t_m}} \\
			& c_{1,1} & c_{2,1} & & c_{m,1} \\
			& \vdots & \vdots & & \vdots \\
			& c_{1,t_1} & c_{2,t_2} & & c_{m,t_m}
		}\]
		where some of the reactions with label $\kappa_{2i}$ might not exist, and for each $1\leq i\leq m$, either $t_i\geq 0$.
	\end{enumerate}
\end{lemma}
\begin{proof}
We start by recalling how to find $\mu_{i,c}$ using a labeled digraph (see proof of Theorem 2 of the electronic supplementary material of \cite{Feliu-Simplifying}). For each non-intermediate complex $c$, consider the labeled digraph $\widehat{G}_c$ with vertex set $\{Y_1,\dots,Y_m,\star\}$ and labeled edges $Y_i\xrightarrow{\k_{Y_i\rightarrow Y_j}} Y_j$ if $Y_i\rightarrow Y_j\in \widetilde{\cR}$,   $\star \xrightarrow{ \k_{c\rightarrow Y_i} x^c} Y_i$ if $c\rightarrow Y_i\in \widetilde{\cR}$ and $Y_i\xrightarrow{\beta_i} \star$ with $\beta_i=\sum_{Y_i\rightarrow c'} \k_{Y_i\rightarrow c'}$ if $\beta_i\neq 0$.

For every vertex $v$ of $\widehat{G}_c$  define $\theta(v)$ as the set of all spanning trees rooted at $v$.\footnote{a spanning tree is rooted at $v$ if $v$ is the only vertex with no outgoing edges} Given such a tree $\tau$, let $\pi(\tau)$ be the product of the labels of the edges of $\tau$. Then 
\begin{equation}\label{eq:mu}
\mu_{i,c}=\frac{\sum_{\tau\in\theta(Y_i)}\pi(\tau)}{\sum_{\tau\in\theta(\star)}\pi(\tau)}.
\end{equation}

\medskip
\noindent
	(i) 
	If one of the reactions is irreversible, then the core network consists of exactly one reaction, either $c\rightarrow c'$ or $c'\rightarrow c$,  and the set $\{\phi_{c\rightarrow c'}(\kappa) \mid c\rightarrow c'\in \cR\}$ is algebraically independent over $\R$.  
	
	If all reactions are reversible, we write
			\[c\ce{<=>[\kappa_1][\kappa_2]}Y_1\ce{<=>[\kappa_3][\kappa_4]}Y_2\ce{<=>[\kappa_5][\kappa_6]}\dots\ce{<=>[\kappa_{2m-1}][\kappa_{2m}]}Y_m\ce{<=>[\kappa_{2m+1}][\kappa_{2m+2}]}c',\]
			and we have $\phi_{c'\rightarrow c}(\kappa)=\kappa_2\mu_{1,c'}$, $\phi_{c\rightarrow c'}(\kappa)=\kappa_{2m+1}\mu_{m,c}$. 
			By the expressions for $\mu_{1,c'}$ and $\mu_{m,c}$ in \eqref{eq:mu},  both rational functions have the same denominator and $\k_{2m+1}$ is not part of their numerator. Therefore, algebraic independence of $\kappa_2\mu_{1,c'}$ and $\kappa_{2m+1}\mu_{m,c}$ follows from the algebraic independence of the numerators  of these two rational functions.  Since $\k_{2m+1}$ is a factor of $\phi_{c\rightarrow c'}(\kappa)$ and is not part of the numerator of $\phi_{c'\rightarrow c}(\kappa)$, the two functions  $\phi_{c\rightarrow c'}(\kappa),\phi_{c'\rightarrow c}(\kappa)$ are algebraically independent over $\R$.

\medskip
\noindent
(ii) We have $\phi_{c_0\rightarrow c_i}(\k) = \ell_i\mu_{m,c_0}$ for $i=1,\dots,p$. Thus the set 
$\{\phi_{c_0\rightarrow c_i}\mid 1\leq i\leq p\}$ is algebraically independent over $\R$ if and only if $\{\ell_i\mid 1\leq i\leq p\}$ is, which clearly holds.

\medskip
\noindent
(iii) The reactions of the core network are of the form $c_0\rightarrow c_{i,j}$. We consider the graph $\widehat{G}_{c_0}$ (removing the edges for which there is no reaction):

\bigskip 		
		
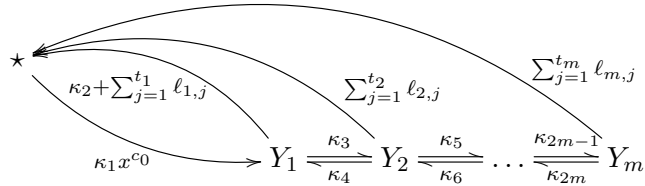
\[\xymatrix{
			\star\ar@/_3ex/[drrr]_{\kappa_{1}x^{c_0}} & & & & & & \\
			& & &  
					Y_1\ar@<+1pt>@{-^>}[r]^{\kappa_3}\ar@/_4ex/[ulll]^{\kappa_2+\sum_{j=1}^{t_1}\ell_{1,j}} & Y_2\ar@<+1pt>@{-^>}[l]^{\kappa_{4}}    \ar@<+1pt>@{-^>}[r]^{\kappa_5}    \ar@/_5ex/[ullll]_(0.15){\sum_{j=1}^{t_2}\ell_{2,j}} & \dots    \ar@<+1pt>@{-^>}[l]^{\kappa_{6}}  \ar@<+1pt>@{-^>}[r]^{\kappa_{2m-1}} & Y_m  \ar@<+1pt>@{-^>}[l]^{\kappa_{2m}}   \ar@/_8ex/[ullllll]_(0.15){\sum_{j=1}^{t_m}\ell_{m,j}}
		}\]
We have $\phi_{c_0\rightarrow c_{i,j}}(\kappa)=\ell_{i,j}\mu_{i,c_0}$. The denominators of the rational functions  $\mu_{i,c_0}$ as given in \eqref{eq:mu} 
agree. Therefore it is enough to check that the polynomials 
 $\rho_{i,j}:=\ell_{i,j}\sum_{\tau\in\theta(Y_i)}\pi(\tau)$ for all $i,j$ are algebraically independent over $\R$.
		
		For each $1\leq i\leq m$, there exists a spanning tree rooted at $Y_i$ involving an edge of the form $Y_j\rightarrow\star$ only for $j\gneqq i$. 
		Now consider the smallest index $i$ such that there exists a complex $c_{i,j}$. The parameter $\ell_{i,j}$ appears  in a polynomial $\rho_{i_1,i_2}$ only for $i_1=i$.  
Hence the polynomials $\rho_{i_1,i_2}$ are algebraically independent if and only if they are for $i_1>i$. We proceed in the same way now considering the smallest index $k>i$ such that there exists a complex $c_{k,j}$. This process terminates in at most $m$ steps.

\end{proof}

Corollary \ref{Corollary Algebraic Independence singletons} and Lemma \ref{Lemma Algebraic Independence more classes}(i) show that typical 
rational functions arising from realistic networks, such as those built from the  mechanism  in Example \ref{Example Popular Intermediates}, fulfil the algebraic independence condition.

	\bigskip

	\section{Another class of species: enzymes} 
		
	In this final section we consider another class of species for which reduction mechanisms have also been defined, namely enzymes, and study how  Gr\"{o}bner  bases of extended and reduced networks relate.

			 \subsection{Enzymes}
A species $E\in\cS$ is an \emph{enzyme} if for every reaction the stoichiometric coefficient of  $E$ in the reactant and the product agree  \cite{Feliu-Persistence}. This automatically gives that the steady state polynomial of $E$ is identically zero, and implies that the concentration of $E$ is constant in time and only depends on the initial amount $e_0$ of $E$.  
For example, $E$ and $F$ are enzymes in the network of Example \ref{Example MAPK mu reduced}.

The core network obtained by removal of $E$ consists of simply removing $E$ from each side of the reaction
 (this is an example of an embedded network, see \cite{AtomsJoshi}). For example, a reaction 
\begin{equation}\label{eq:enzyme}
 X_1 + E \xrightarrow{\k_1} X_2 + E \qquad \textrm{becomes } \qquad X_1\xrightarrow{k_1} X_2.
 \end{equation}
After fixing the initial amount of enzyme $e_0$,  the steady states of the extended network satisfying that the concentration of $E$ is $e_0$ agree with the steady states of the core network with $k_1= e_0 \k_1$.

This might lead one to think that enzymes are  redundant and that similar properties as those that hold for  intermediates also hold for enzymes. For example, one might think there is an easy way to obtain a Gr\"{o}bner basis of the steady state ideal of the extended network from one of the core network, or that a binomial steady state ideal remains binomial upon removal of intermediates. But this is not the case, as the following examples illustrate.

	\begin{example}\label{Example Simple Enzyme Destroying Binomiality}
	Let $\cN$ be the  network
	
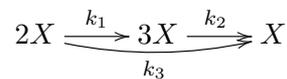
\[\xymatrix{2X\ar[r]^{k_1}\ar@/_/[rr]_{k_3} & 3X\ar[r]^{k_2} & X}\]
	A binomial basis of the steady state ideal is $\{-2k_2x^3+(k_1-k_3)x^2\}$. Now consdier the following network by adding one enzyme $E$:
	\[2X \ce{->[\kappa_1]} 3X \ce{->[\k_2]} X \qquad 2X + E \ce{->[\k_3]} X + E.\]
	A reduced Gr\"{o}bner basis of its steady state ideal is $\{x^3-\tfrac{\k_1}{2\k_2}x^2+\tfrac{\k_3}{2\k_2}x^2e\}$, and hence this ideal is not  binomial.
	\end{example}
	
	The previous example suggests the following: Consider a reaction as in \eqref{eq:enzyme}. One might obtain a Gr\"{o}bner basis of the steady state ideal of the extended network by considering a Gr\"{o}bner basis of the steady state ideal of the core network and substituting the parameter  $\k_1$ by $k_1 e$.
	The following example gives a negative answer to this question.
	
	\begin{example}\label{Example Enzymes As Coefficients}
	Let $\cN$ be the following network
	
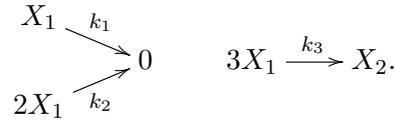
\[
	\xymatrix @C=1.75pc @R=0.1pc { X_1\ar[rd]^{k_1} & & & \\
		& 0 & 3X_1\ar[r]^{k_3} & X_2.\\
		2X_1\ar[ur]_{k_2} & & &
		}\]
	The set of steady state polynomials is 
	\[\{-k_1x_1-2k_2x_1^2-3k_3x_1^3,\,k_3x_1^3\}.\]
	With every arbitrary monomial order on $\mathbb{R}(k)[x]$, the reduced Gr\"{o}bner basis of the steady state ideal is $\{x_1\}$.
	
	Let now $\cN'$ be the extension of $\cN$ via the  enzyme $E$:
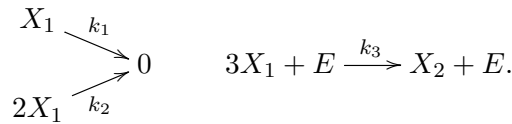
	\[
	\xymatrix @C=1.75pc @R=0.1pc { X_1\ar[rd]^{k_1} & & & \\
		& 0 & 3X_1+E\ar[r]^{k_3} & X_2+E.\\
		2X_1\ar[ur]_{k_2} & & &
	}\]
	The set of steady state polynomials of $\cN'$ is
{\small \[	\{-\k_1 x_1-2\k_2 x_1^2-3\k_3x_1e,\,\k_3 x_1^3e\}.
\]}
	The steady state ideal is  different from  $\langle x_1\rangle$. Thus, there is not a monomial order on $\mathbb{R}(\k)[x,e]$ for which the reduced Gr\"{o}bner basis can be obtained from the set $\{x_1\}$ by making the substitution $k_3=\k_3e$.

	\end{example}

	\begin{example}\label{Example Enzymes Linear Binomiality}
	When a binomial basis of the steady state ideal is obtained from linear combinations of the steady state polynomials (see \cite{Conradi-Binomial}), then the steady state ideal of the core network is binomial if and only if that of the extended network is.
	\end{example}

	\subsection*{Acknowledgements}
This work has been supported by the Danish Research Council for Independent Research. We thank Martin Helmer, Ang\'elica Torres and Carsten Wiuf for comments on previous versions of this manuscript.


	\end{document}